\documentclass[twocolumn,10pt]{IEEEtran}

\usepackage{pgfplots}
\usetikzlibrary{backgrounds,automata}
\usetikzlibrary{plotmarks}
\usetikzlibrary{patterns}
\usetikzlibrary{calc,positioning,fit,backgrounds}
\usetikzlibrary{shapes,snakes}
\usetikzlibrary{intersections,positioning}
\usepackage{caption}
\usepackage{cite}
\usepackage{amsmath,amssymb,amsfonts,color,soul,amsthm}
\usepackage{algorithmic}
\usepackage{graphicx}
\usepackage{textcomp}
\usepackage{xcolor}

\def\argmax{\operatorname{arg~max}}
\def\E{\mathbb{E}}

\def\ie{{\em i.e.}}

{}

\newtheorem{theorem}{Theorem}
\newtheorem{prop}{Proposition}

\usepackage[caption=false,font=footnotesize]{subfig}
\def\BibTeX{{\rm B\kern-.05em{\sc i\kern-.025em b}\kern-.08em
    T\kern-.1667em\lower.7ex\hbox{E}\kern-.125emX}}

\begin{document}

\title{Dimensioning an Indoor SISO RIS-system: Approximations and Equivalence Models\\
}

\author{
\IEEEauthorblockN{Saksham Bhushan, Sai Teja Suggala, Arzad A. Kherani, Sreejith T. V.}\\
\IEEEauthorblockA{\textit{Dept. of Electrical Engineering and Computer Science \\
Indian Institute of Technology Bhilai, India }\\
Email: \{sakshamb, suggalat, arzad.alam, sreejith\}@iitbhilai.ac.in}
}


\maketitle

\begin{abstract}

We provide closed-form approximations to the performance gain achieved in a RIS-assisted communication. We then consider a network deployment of RIS and Transmitter-Receiver pairs and use these approximate expressions to provide equivalence models which state that the performance of a RIS-equipped network is similar to the performance of an appropriately spatially scaled system. We provide a way of assigning available RIS to assist communication between several transmitter-receiver pairs. Several such approximations are expected to spawn from this study, with more clarity on structural aspects of the gains achieved in RIS-assisted communications.

\end{abstract}


\section{Introduction}


Reconfigurable Intelligent Surfaces (RIS) are man-made surfaces with small passive electromagnetic structures which can be controlled electronically through a controller to achieve unconventional results.
%
%
RIS is being compared to some breakthrough concepts in the field of wireless communication \cite{ris_how, towards}. However, for RIS to be impactful in future wireless networks, it needs to seamlessly integrate with existing wireless communication technology \cite{an_idea} and how it outperforms existing similar technologies such as relay \cite{beat_relay, compare_relay}.
Some interesting applications of RIS include active beamforming at the access point and passive beamforming of RIS elements to enhance the received power \cite{act_pass_beam}, another one being \cite{ap_trans} where an umodulated carrier wave is impinged upon the RIS and the phase shifters work as a signal modulator.

Various phase estimation and optimization algorithms have been proposed  for enhanced received power and eventually channel capacity \cite{achieve_max, dl_phase, low_comp}. Although, it is not always possible to perfectly estimate the phase, \cite{phase_error} does performance analysis on such a system. The reliability of the reconfigurability of RIS is examined in \cite{error_model} by proposing an error model and a general methodology for error analysis but this is out of scope of this paper. In \cite{power_scaling}, the reliability problem in programmable metamaterials analysed with error model and potential impacts of faults are studied. The RIS placement, number of elements, tilt or rotation are considered in \cite{2021_performance} using simulations.


In this paper, we use a corridor model for RIS and Tx-Rx deployment and we further our discussion with approximating the system for improvement in system throughput via a RIS-assisted channel and also provide a scalable model for this approximation. We validated these approximations, equivalence models and scaling laws using SimRIS simulator \cite{basar2020SimRIS}.

\section{Network Scenario and Signal Propagation Model}
\label{sec:netk_scenario}

We consider a scenario where a long corridor has 
multiple RIS elements located at random locations on the ceiling. Transmitters and receivers are distributed randomly with uniform density across the corridor. With $N_t$ transmit antenna,  $N_r$ receive antenna and $N$ RIS elements, the channel matrix $C$ in case of RIS-assisted system as given in~\cite{basar2020SimRIS}.
\begin{equation}
C(\Phi)=G\Phi H+D,\label{eq_C}
\end{equation}
where $H \in \mathbb{C}^{N \times N_t}$ is the channel coefficient matrix between the transmitter and RIS, $G \in \mathbb{C}^{N_r \times N}$ is the channel coefficient matrix between the RIS and the receiver, and $D \in \mathbb{C}^{N_r \times N_t}$ is the channel coefficient matrix for the direct channel between the transmitter and receiver. The $\Phi$ is a diagonal matrix which represents the reconfigurable responses of the RIS elements and optimizing this matrix will enable RIS to improve the system performance, \ie,\, the achievable rate (R) for the system is given in \cite{basar2020SimRIS_3} as 
\begin{equation}
R=\max_{\Phi} \log_2\Big|I_{Nr}+\frac{P_t}{\sigma^2}C^\dagger (\Phi)C(\Phi)\Big| \mbox{\qquad bits/s/Hz}.\label{eq_rate}
\end{equation}
where $C^\dagger$ is the complex conjugate transpose of $C$ matrix, $P_t$ is the transmitted power in Watts, $\sigma^2$ is the average noise power. For our system the noise power is assumed to be -100 dBm and the transmission power is set at 30 dBm.
We are considering a SISO system, $N_t = N_r = 1$. By considering a SISO system it is easier to realize the exact phase equalization for the system whereas phase equalization in MIMO systems is restricted by an optimization problem.


\section{RIS Approximation Model}

We start with a single RIS element scenario and extend it to multiple RIS elements distributed randomly over the ceiling of the service area. Consider the scenario shown in Fig.~\ref{fig_setup1} where a RIS element is shown deployed at a height $z$ and at a horizontal distance $y$ from the transmitter when the distance between the transmitter and the receiver is $r$.

\begin{figure}[htbp]
\centerline{\includegraphics[width=8cm]{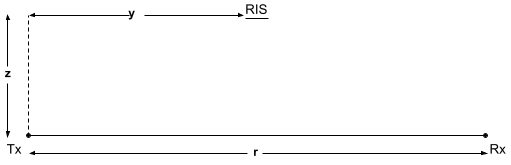}}
\caption{Deployment setup for approximating the RIS-assisted model.}
\label{fig_setup1}
\end{figure}

%




Assume that the transmitter is equipped with only one transmit antenna (extension to multiple antenna is straightforward, but complicates the equations used). 
The total power received at the receiver in case of such RIS-assisted system is given by
\begin{equation}
P(r,z;y)=P_{LOS}(r)+P_{RIS}(r,z;y)\label{eq_T},
\end{equation}
where $P_{LOS}(r)$ is the received power at the receiver only through the direct line of sight link between the transmitter and the receiver and $P_{RIS}(r,z;y)$ is the \textit{additional} power only due to the \emph{scattering} of the EM-waves through RIS. Note that, we implicitly assumed that the two components of power are independent and the receiver is able to combine the two signals optimally while compensating for the difference in phases.
The \emph{expected} received power via the RIS element can be represented as
\begin{equation}
P_{RIS}(r,z;y)=\frac{S}{(y^2+z^2)((r-y)^2+z^2)}, \label{eq:P_ris_avg}
\end{equation}
where $S$ is a constant for a given RIS element, transmitter-receiver pair and depends on the respective antenna gains and scattering profile. The randomness introduced by the environment (in terms of fading) is incorporated in $S$ as we are looking at average values).


\noindent\textit{Scaling property and an approximation:} Now we can see some scaling properties of RIS assisted systems. We can see that $P_{RIS}(r,z;y)$ can be normalized to obtain a scaled version of \eqref{eq:P_ris_avg} with respect to $r$, by substituting $p=\frac{y}{r}$ and $q=\frac{z}{r}$.\ie,
\begin{align}
P_{RIS}(r,z;y)&=\left(\frac{q}{z}\right)^4P_{RIS}(1,q;p) =\frac{1}{r^4}\widehat{P}_{RIS}^{Av}(\frac{z}{r};\frac{y}{r})
\end{align}	
where $\widehat{P}_{RIS}^{Av}(q;p)$ is the additional average received power of normalized system. Thus the average power of the original system can be obtained by scaling the normalized system, which is independent of $r$ and hence, one needs to study only the function $\widehat{P}_{RIS}^{Av}(q;p)$, $p\in [0,1]$, $q\geq 0$. Now we can see an interesting structure which proves that using a value of $q>\frac{1}{2}$ does not provide any promising return from the RIS deployment.

The extremum points of the function $\widehat{P}_{RIS}^{Av}(q;p)$ for a given value of $q$ are at $p=1/2$ and $p=\frac{1}{2}(1\pm\sqrt{1 - 4q^2})$. This function can be approximated using linear piecewise function as follows (for $q<1/2$):
%


\begin{equation}\label{eq:piecewise}
\frac{\widehat{P}_{RIS}^{Av}(q;p)}{S} = \begin{cases}
\frac{2p}{(1+q^2)(1-\sqrt{1-4q^2})}+\frac{1}{q^2(q^2+1)};\\ \qquad 0\leq p\leq \frac{1}{2}(1-\sqrt{1 - 4q^2})\\
\frac{2(16q^4-8q^2+1)(1/2-p)}{q^2\sqrt{1-4q^2}(4q^2+1)^2}+\frac{16}{(4q^2+1)^2}; \\
\qquad \frac{1}{2}(1-\sqrt{1 - 4q^2}) < p \leq \frac{1}{2},
\end{cases}
\end{equation}

and, similarly for $p>1/2$ using the symmetry in the piecewise approximation, we have,
\begin{equation}
\widehat{P}_{RIS}^{Av}(q;p)=\widehat{P}_{RIS}^{Av}(q;1-p) \label{eq:symmetry}
\end{equation}

The approximation in \eqref{eq:piecewise} can be used, when we consider multiple RIS system where the RIS locations follow some distributions, for which expectations over $r, y$ are to be taken, those may not be simplified to closed form expressions. We can approximate the integration of \eqref{eq:piecewise} as sum of areas of triangles with this piecewise linear approximations. Similar approximations can be derived for $q\geq 1/2$ also, but as the additional power improvement is negligible in this case and we do not consider this case. 


In the next section, we will consider the scenario where multiple RISs are deployed on the ceiling. 
\section{Network with Random RIS Deployment}

Assume that the RIS are distributed along the corridor according to an independent Poisson process of intensity $\rho$.
The average additional received power can be obtained   by Campbell's theorem\footnote{Campbell's theorem: $\E[\sum_{x\in\Phi}f(x)]=\beta\int_{\mathbb{R}^d }f(x)dx$.} \cite{stoyan2013stochastic} and we have the following theorem. 
\begin{theorem} The average additional received power at a distance $r$ from transmitter, obtained by deploying Poisson distributed RIS elements with density $\rho$ on a ceiling of height $z$ m from transmit-receiver plane is given by,
	\begin{align}P_{RIS}^{Av}(r,z)&=2G\rho\frac{ \log \left(\left(\frac{r^2 \left(\frac{z^2}{r^2}+1\right)}{z^2}\right)^{\frac{z}{r}}\right)+\tan
		^{-1}\left(\frac{r}{z}\right)}{r^3 \left(\frac{4 z^3}{r^3}+\frac{z}{r}\right)}
	\label{eq:Pav-theorem}
	\end{align}
\end{theorem}
\begin{proof} Using \eqref{eq:P_ris_avg} and Campbell's theorem.	
%
\end{proof}
Here we have not considered RIS elements beyond the reciever $(y>r)$, as the reflections from those elements will be negligible as the total path lengths of those cases are higher than the elements of $(0,r)$. 

This theorem can be extended to multidimensional distribution and the next theory provides us an important scaling law for RIS assisted system. Let us assume RIS elements are distributed on the wall with density $\beta$, then the average additional power can be obtained using the following theorem. 

\begin{theorem}
Average additional RIS-assisted received power     \begin{equation}
        P_{RIS}^{Av}=\frac{\beta}{r^2}\int_{q} P_{RIS}^{Av}(1,q) dq
    \label{eq:theorem_2}
    \end{equation}
decreases as $\frac{1}{r^2}$, with $\int_{q} P_{RIS}^{Av}(1,q) dq$ being independent of $r$.
\end{theorem}
We can see that by using the scaling law, this additional power can be obtained with integral in which variables are multiplicatively separable.

Next proposition provide an approximation for the average additional received power given in \eqref{eq:Pav-theorem}, by using our piecewise linear approximation model given in \eqref{eq:piecewise}.
\begin{prop}The approximated average additional received power in normalized system with Poisson distributed RIS elements is given by
	\begin{align}
\frac{P_{RIS}^{Av}(r,qr)}{S}&\approx\frac{\rho}{r^3}\left(\frac{(2+q^2)(1-\sqrt{1-4q^2})}{2q^2(1+q^2)}\right.\nonumber\\
& + \left.\frac{(16q^4+24q^2+1)\sqrt{1-4q^2}}{2q^2(1+4q^2)^2}\right).
	\label{eq:int_piece}
	\end{align}
\end{prop}
%





In Fig. \ref{fig:tadd_for_varying_y2}, we evaluate the achievable rate for different positions of RIS operating in indoor environment using the SimRIS channel simulator. For this experiment, we consider $r = 50$ meters, which means that the Tx and Rx are fixed at a distance $50$ m apart, $z$ which is the perpendicular distance of RIS from the line connecting Tx and Rx, and for each valus of z, we vary y which is the distance of RIS from the origin along the horizontal axis and examine the effect of moving the RIS between Tx and Rx at a perpendicular distance. Upon analysing we observe two peaks in achievable rate which suggests that the ideal position of RIS placement is nearer to the transmitter or receiver. As we can see from the Fig. \ref{fig:tadd_for_varying_y2} that the RIS can provide a remarkable improvement in the achievable rate. The minimum improvement in achievable rate is $1$ unit which accounts for approximately $15\%$ improvement in this case.
\begin{figure}[htbp]
\centerline{\includegraphics[width=7cm]{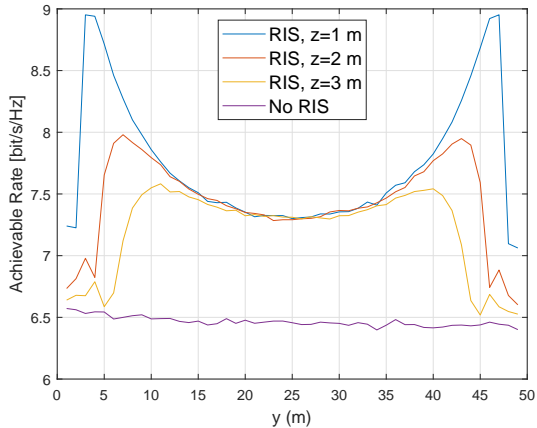}}
\caption{Achievable rates of RIS-assisted system and No-RIS system in bit/sec/Hz under varying value of y with z = 1 m, z = 2 m and z = 3 m. The results are obtained using SimRIS simulator.}
\label{fig:tadd_for_varying_y2}
\end{figure}

 From \eqref{eq:piecewise}, it can be observed that for $q<1/2$ the function has $1$ minimum at $p=1/2$ and 2 maxima whereas for $q>1/2$ the function has only one maxima at $p=1/2$. This structural behaviour of the $P_{RIS}$ function can be observed in results from SimRIS with scatterers (Fig.~\ref{fig:tadd_for_varying_y2}). As we decrease the value of $q$ the two maxima move towards the two end points, thus if we place RIS nearer to the vertical plane on which transmitter and receiver lie, the value of $P_{RIS}^{Av}$ is greater nearer to the end-points which suggests that optimal position for RIS is near to the transmitter and the receiver. Hence, in the remaining numerical results presented in the paper, we concentrate on the region $q<\frac{1}{2}$.






\begin{figure}[htbp]
\centerline{\includegraphics[width=7cm]{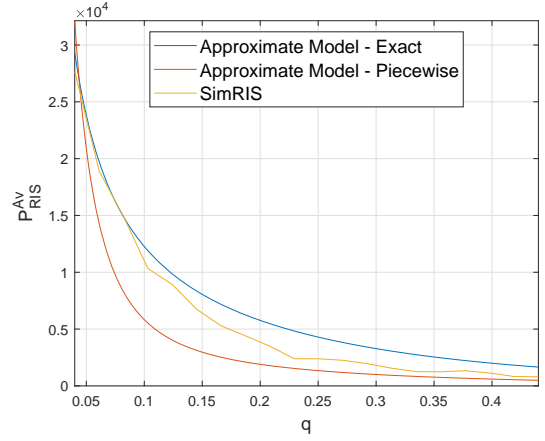}}
\caption{Comparison of results from SimRIS, actual approximation and the piecewise approximation.}
\label{fig:all3}
\end{figure}

Fig. \ref{fig:all3} gives the $P^{Av}_{RIS}(r,z)$ from \eqref{eq:Pav-theorem}, 
the piecewise approximation given in \eqref{eq:int_piece}, and the results from SimRIS. The simulation assumes $\rho=1$ RIS/meter, each RIS had $256$ RIS elements. The value of $S$ was back-calculated from that of simulation for $q\approx 0$. The plot demonstrates that the piecewise approximation provides reliable performance improvement results for smaller values of $q$. The match between approximation and exact can be improved by extending the piecewise approximation, but that exercise is beyond the scope of this paper. 

 \definecolor{mycolor1}{rgb}{0.00000,0.44700,0.74100}%
 \definecolor{mycolor2}{rgb}{0.85000,0.32500,0.09800}%

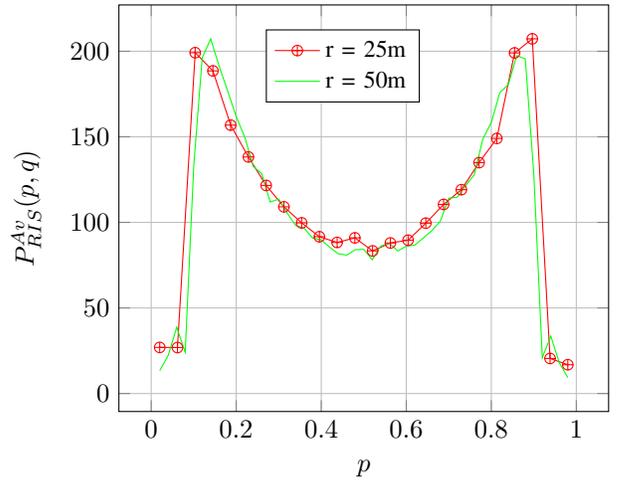
\begin{figure}[h]
	\centering
	\begin{tikzpicture}
	\begin{axis}[scale=0.95,
	grid = both,
	legend style={font=\small,
	at={(0.3,.85)},anchor=west,
	},
	xlabel ={$p$},
	ylabel = {$P^{Av}_{RIS}(p,q)$}
	]
	
	\addplot [color=mycolor2,mark=oplus,color=red]
	table[row sep=crcr]{%
		0.02	26.9510878141508\\
		0.0617391304347826	26.8634151166129\\
		0.103478260869565	199.172401236612\\
		0.145217391304348	188.524208268752\\
		0.18695652173913	156.83381817788\\
		0.228695652173913	138.246716616585\\
		0.270434782608696	121.596056577965\\
		0.312173913043478	109.091192723485\\
		0.353913043478261	99.7191046886975\\
		0.395652173913044	91.6300671742437\\
		0.437391304347826	88.292091496021\\
		0.479130434782609	90.9578079407022\\
		0.520869565217391	83.3373915629672\\
		0.562608695652174	87.9752660143608\\
		0.604347826086956	89.6120411493479\\
		0.646086956521739	99.5649896978014\\
		0.687826086956522	110.496311758975\\
		0.729565217391304	119.108066090375\\
		0.771304347826087	134.928203319915\\
		0.81304347826087	149.050053784925\\
		0.854782608695652	199.02604410764\\
		0.896521739130435	207.24838371891\\
		0.938260869565217	20.4929623632637\\
		0.98	16.8356036924734\\
	};
	\addlegendentry{r = 25m}
	\addplot [color=mycolor1,mark=triangle,no marks,color=green]
	table[row sep=crcr]{%
		0.02	13.3593138251599\\
		0.04	22.2939823010302\\
		0.06	38.7362756377499\\
		0.08	24.12285127684\\
		0.1	132.63011770562\\
		0.12	196.1815124399\\
		0.14	207.24838371891\\
		0.16	191.09604925329\\
		0.18	176.68303256828\\
		0.2	161.84174239235\\
		0.22	149.8346405165\\
		0.24	132.6983429072\\
		0.26	128.45772967942\\
		0.28	111.87773000357\\
		0.3	113.63891427659\\
		0.32	105.43363814829\\
		0.34	98.20364261347\\
		0.36	96.5379818475099\\
		0.38	90.83112898584\\
		0.4	89.7900989786299\\
		0.42	85.44921730485\\
		0.44	81.6544271323\\
		0.46	80.9109154119101\\
		0.48	84.0239981518602\\
		0.5	84.2305835334598\\
		0.52	78.12157715988\\
		0.54	86.15408107345\\
		0.56	87.8599861831499\\
		0.58	83.3128574719001\\
		0.6	86.11746251155\\
		0.62	86.7061333373501\\
		0.64	90.89913725933\\
		0.66	95.12123980004\\
		0.68	100.67032959882\\
		0.7	114.19549563331\\
		0.72	114.71870866416\\
		0.74	121.63605047886\\
		0.76	127.87592608528\\
		0.78	148.61400602132\\
		0.8	158.30396804009\\
		0.82	175.75802336\\
		0.84	180.52642539936\\
		0.86	197.40669823427\\
		0.88	195.55169015607\\
		0.9	128.51832811202\\
		0.92	20.95797166131\\
		0.94	33.55741771972\\
		0.96	18.0219810889901\\
		0.98	9.29721558892993\\
	};
	\addlegendentry{r = 50m}
	\end{axis}
	\end{tikzpicture}
	\caption{Illustration of the equivalence of scaled systems. Plot of $P_{RIS}$ for normalized system with varying value of $p$. }.
	\label{fig:equivalence}
\end{figure}
%
Fig.~\ref{fig:equivalence} shows that the results for different combinations of $y$, $z$ and $r$ depend only on the values of $p = \frac{y}{r}$ and $q=\frac{z}{r}$, up to an appropriate multiplicative constant. This is a considerable simplification as we now have the basic invariants of the system.
\section{RIS Assignment Algorithm}
Assuming a noise limited scenario, when there are many users i.e., Tx-Rx pairs, then there must be a scheduler that has to assign RIS to a particular pair (optimizing $\Phi$ for some of the available RIS for that Tx-Rx pair). The scheduler must take inputs of locations of Tx, Rx and available RIS. For this, we propose \eqref{eq:piecewise} to assign only those RIS elements that can achieve higher throughput as in Fig \ref{fig:tadd_for_varying_y2}.

Let us assume that we wish to ensure that the additional power received is independent of $z$ for a given value of $r$, denote this additional power as $P^*$. Then, we would like to assign the RIS close to $q\approx 0$ to support this communication. Thus, for a given $q$, from simulations $P_{min}(q) = P(0.5,q)$ and if $P_{max}(q) = \max_p P(p,q)$ and $p^*(q) = \argmax_{p < \frac{1}{2}} P(p,q)$ then we seek the level sets of $P(p,q)$ such that 
if we require 
\[P({p^*(q)-\Delta_l},q) = P({p^*(q)+\Delta_r},q)\]
then the total additional power is $P^*$, i.e., $\int_{p={p^*(q)-\Delta_l}}^{{p^*(q)+\Delta_r}} = P^*$.
We denote $\Delta=\Delta_l+\Delta_r$, and 
\begin{equation}
\begin{aligned}
x^*(q) & = \frac{P({p^*(q)-\Delta_l},q) - P_{min}(q)}{P_{max}(q) - P_{min}(q)}  \\
& = \frac{P({p^*(q)+\Delta_r},q) - P_{min}(q)}{P_{max}(q) - P_{min}(q)}.
\end{aligned}
\end{equation}


Fig.~\ref{fig:delta_percent} provides the value of $x^*(q)$ as a function of $q$ from SimRIS, the piecewise linear approximation \eqref{eq:piecewise}, and \eqref{eq:P_ris_avg}. For $q=0.04$, corresponding to $z= 1$m for $r=25$m, the value of $x^*(0.04)=0.85$ was set and the corresponding value of $P^*$ was obtained. This value of $P^*$ was used to get the $x^*(q)$ for the other values of $q$. The it is seen that $x^*(q)$ decreases with $q$ almost linearly. For the same approach, the values of $\Delta$ as a function of $q$ are provided in Fig.~\ref{fig:delta_base}.
\begin{figure}[htbp]
\centerline{\includegraphics[width=7cm]{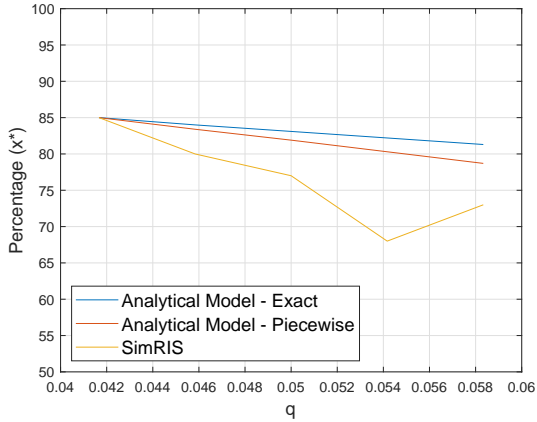}}
\caption{Figure providing the behavior of $x^*(q)$ from SimRIS, Piecewise linear approximation and the exact analytical model.}
\label{fig:delta_percent}
\end{figure}
%






\begin{figure}[htbp]
\centerline{\includegraphics[width=7cm]{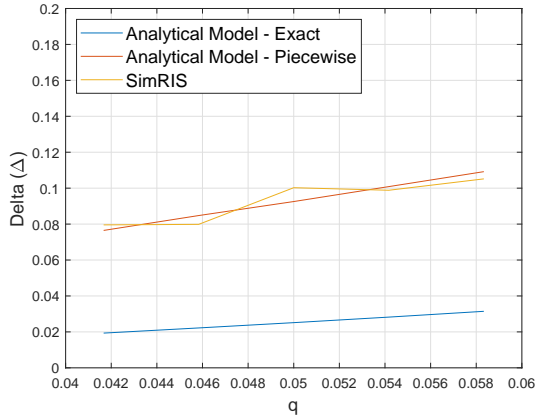}}
\caption{Figure providing the behavior of $\Delta$ from SimRIS, Piecewise linear approximation and the exact analytical model.}
\label{fig:delta_base}
\end{figure}






This gives an interesting phenomenon. We can see that smaller values of $q$ provide higher maximum additional throughput, but the value of $\Delta$, i.e., the usable span, also reduces significantly. This means that, when there is a high enough density of Tx-Rx pairs, we may want to have a small value of $q$ so that the smaller value of $\Delta$ can be utilized for multiple Tx-Rx pairs.







\section{Conclusions}

We considered a RIS-assisted communication and provided an approximation for the improvement. We provided structural results and a way of using these approximations in assigning RIS to individual Tx-Rx pairs. We believe that several such approximations can be derived now and online RIS assignment algorithms can be implemented.



\bibliographystyle{IEEEtran}
\bibliography{main}

\begin{thebibliography}{10}
\providecommand{\url}[1]{#1}
\csname url@samestyle\endcsname
\providecommand{\newblock}{\relax}
\providecommand{\bibinfo}[2]{#2}
\providecommand{\BIBentrySTDinterwordspacing}{\spaceskip=0pt\relax}
\providecommand{\BIBentryALTinterwordstretchfactor}{4}
\providecommand{\BIBentryALTinterwordspacing}{\spaceskip=\fontdimen2\font plus
\BIBentryALTinterwordstretchfactor\fontdimen3\font minus
  \fontdimen4\font\relax}
\providecommand{\BIBforeignlanguage}[2]{{%
\expandafter\ifx\csname l@#1\endcsname\relax
\typeout{** WARNING: IEEEtran.bst: No hyphenation pattern has been}%
\typeout{** loaded for the language `#1'. Using the pattern for}%
\typeout{** the default language instead.}%
\else
\language=\csname l@#1\endcsname
\fi
#2}}
\providecommand{\BIBdecl}{\relax}
\BIBdecl

\bibitem{ris_how}
M.~{Di Renzo}, A.~{Zappone}, M.~{Debbah}, M.~S. {Alouini}, C.~{Yuen}, J.~{de
  Rosny}, and S.~{Tretyakov}, ``Smart radio environments empowered by
  reconfigurable intelligent surfaces: How it works, state of research, and the
  road ahead,'' \emph{IEEE Journal on Selected Areas in Communications},
  vol.~38, no.~11, pp. 2450--2525, 2020.

\bibitem{towards}
Q.~{Wu} and R.~{Zhang}, ``Towards smart and reconfigurable environment:
  Intelligent reflecting surface aided wireless network,'' \emph{IEEE
  Communications Magazine}, vol.~58, no.~1, pp. 106--112, 2020.

\bibitem{an_idea}
M.~D. Renzo, M.~Debbah, D.-T. Phan-Huy, A.~Zappone, M.-S. Alouini, C.~Yuen,
  V.~Sciancalepore, G.~C. Alexandropoulos, J.~Hoydis, H.~Gacanin, J.~de~Rosny,
  A.~Bounceu, G.~Lerosey, and M.~Fink, ``Smart radio environments empowered by
  ai reconfigurable meta-surfaces: An idea whose time has come,'' 2019.

\bibitem{beat_relay}
E.~{Björnson}, .~{Özdogan}, and E.~G. {Larsson}, ``Intelligent reflecting
  surface versus decode-and-forward: How large surfaces are needed to beat
  relaying?'' \emph{IEEE Wireless Communications Letters}, vol.~9, no.~2, pp.
  244--248, 2020.

\bibitem{compare_relay}
A.~A. {Boulogeorgos} and A.~{Alexiou}, ``Performance analysis of reconfigurable
  intelligent surface-assisted wireless systems and comparison with relaying,''
  \emph{IEEE Access}, vol.~8, pp. 94\,463--94\,483, 2020.

\bibitem{act_pass_beam}
Q.~{Wu} and R.~{Zhang}, ``Intelligent reflecting surface enhanced wireless
  network: Joint active and passive beamforming design,'' in \emph{2018 IEEE
  Global Communications Conference (GLOBECOM)}, 2018, pp. 1--6.

\bibitem{ap_trans}
E.~Basar, ``Transmission through large intelligent surfaces: A new frontier in
  wireless communications,'' 2019.

\bibitem{achieve_max}
C.~{Huang}, A.~{Zappone}, M.~{Debbah}, and C.~{Yuen}, ``Achievable rate
  maximization by passive intelligent mirrors,'' in \emph{2018 IEEE
  International Conference on Acoustics, Speech and Signal Processing
  (ICASSP)}, 2018, pp. 3714--3718.

\bibitem{dl_phase}
B.~{Sheen}, J.~{Yang}, X.~{Feng}, and M.~M.~U. {Chowdhury}, ``A deep learning
  based modeling of reconfigurable intelligent surface assisted wireless
  communications for phase shift configuration,'' \emph{IEEE Open Journal of
  the Communications Society}, vol.~2, pp. 262--272, 2021.

\bibitem{low_comp}
Z.~{Yigit}, E.~{Basar}, and I.~{Altunbas}, ``Low complexity adaptation for
  reconfigurable intelligent surface-based mimo systems,'' \emph{IEEE
  Communications Letters}, vol.~24, no.~12, pp. 2946--2950, 2020.

\bibitem{phase_error}
M.~{Badiu} and J.~P. {Coon}, ``Communication through a large reflecting surface
  with phase errors,'' \emph{IEEE Wireless Communications Letters}, vol.~9,
  no.~2, pp. 184--188, 2020.

\bibitem{error_model}
H.~{Taghvaee}, A.~{Cabellos-Aparicio}, J.~{Georgiou}, and S.~{Abadal}, ``Error
  analysis of programmable metasurfaces for beam steering,'' \emph{IEEE Journal
  on Emerging and Selected Topics in Circuits and Systems}, vol.~10, no.~1, pp.
  62--74, 2020.

\bibitem{power_scaling}
\BIBentryALTinterwordspacing
E.~Bjornson and L.~Sanguinetti, ``Power scaling laws and near-field behaviors
  of massive mimo and intelligent reflecting surfaces,'' \emph{IEEE Open
  Journal of the Communications Society}, vol.~1, p. 1306–1324, 2020.
  [Online]. Available: \url{http://dx.doi.org/10.1109/OJCOMS.2020.3020925}
\BIBentrySTDinterwordspacing

\bibitem{2021_performance}
M.~Toumi and A.~Aijaz, ``System performance insights into design of
  ris-assisted smart radio environments for 6g,'' 2021.

\bibitem{basar2020SimRIS}
\BIBentryALTinterwordspacing
E.~{Basar} and I.~{Yildirim}, ``{SimRIS} channel simulator for reconfigurable
  intelligent surface-empowered communication systems,'' May 2020. [Online].
  Available: \url{https://arxiv.org/abs/2006.00468}
\BIBentrySTDinterwordspacing

\bibitem{basar2020SimRIS_3}
\BIBentryALTinterwordspacing
E.~Basar and I.~Yildirim, ``Reconfigurable intelligent surfaces for future
  wireless networks: A channel modeling perspective,'' March 2021. [Online].
  Available: \url{https://arxiv.org/abs/2008.01448}
\BIBentrySTDinterwordspacing

\bibitem{stoyan2013stochastic}
S.~Chiu, D.~Stoyan, W.~Kendall, and J.~Mecke, \emph{Stochastic Geometry and Its
  Applications}, ser. Wiley Series in Probability and Statistics.\hskip 1em
  plus 0.5em minus 0.4em\relax Wiley, 2013.

\end{thebibliography}

\end{document}